\documentclass[a4paper,UKenglish,cleveref, autoref, thm-restate]{lipics-v2021}

\pdfoutput=1 
\hideLIPIcs  

\definecolor{keywordcolor}{rgb}{0.7, 0.1, 0.1}   
\definecolor{tacticcolor}{rgb}{0.0, 0.1, 0.3}    
\definecolor{commentcolor}{rgb}{0.4, 0.4, 0.4}   
\definecolor{stringcolor}{rgb}{0.5, 0.3, 0.2}    
\definecolor{symbolcolor}{rgb}{0.1, 0.2, 0.7}    
\definecolor{sortcolor}{rgb}{0.1, 0.5, 0.1}      
\definecolor{attributecolor}{rgb}{0.7, 0.1, 0.1} 
\definecolor{errorcolor}{rgb}{1, 0, 0}           

\usepackage{scalefnt}
\usepackage{upgreek}
\usepackage{xspace}
\usepackage{fontawesome}
\usepackage{tikz}

\DeclareMathOperator{\round}{round}
\DeclareMathOperator{\Prop}{Prop}
\DeclareMathOperator{\Th}{Th}
\DeclareMathOperator{\denom}{denom}

\newcommand{\N}{\mathbb{N}}
\newcommand{\Z}{\mathbb{Z}}
\newcommand{\Q}{\mathbb{Q}}
\newcommand{\R}{\mathbb{R}}
\newcommand{\Circ}{\mathbb{S}}
\newcommand{\mathlib}{\texttt{mathlib}\xspace}
\newcommand{\link}{\ensuremath{{}^\text{\faExternalLink}}}
\newcommand{\mllink}[1]{\href{https://github.com/leanprover-community/mathlib/blob/9956c3806d0f9553e5c6e6af68970563a1619cd1/src/#1}{\link}}

\title{A formalisation of Gallagher's ergodic theorem} 

\author{Oliver Nash}{Imperial College, London, United Kingdom \and \url{http://olivernash.org}}{o.nash@imperial.ac.uk}{https://orcid.org/0000-0001-7208-6307}{} 

\authorrunning{O. Nash} 

\Copyright{Oliver Nash} 



\category{} 

\relatedversion{} 

\supplement{The formalisation is available in the master branch of Lean's Mathematical Library, \mathlib.}
\supplementdetails{Software (Source Code)}{https://github.com/leanprover-community/mathlib}

\acknowledgements{It is a pleasure to thank Andrew Pollington who suggested this project
during the conference Lean for the Curious Mathematician, held at Brown University (ICERM) in 2022.
I also wish to thank Anatole Dedecker, Heather Macbeth, Patrick Massot, and Junyan Xu, all of whom
were of direct assistance. Lastly I especially wish to thank Sébastien Gouëzel for many helpful
suggestions and Kevin Buzzard for many useful conversations.}

\nolinenumbers 

\EventEditors{John Q. Open and Joan R. Access}
\EventNoEds{2}
\EventLongTitle{42nd Conference on Very Important Topics (CVIT 2016)}
\EventShortTitle{CVIT 2016}
\EventAcronym{CVIT}
\EventYear{2016}
\EventDate{December 24--27, 2016}
\EventLocation{Little Whinging, United Kingdom}
\EventLogo{}
\SeriesVolume{42}
\ArticleNo{23}

\begin{document}

\maketitle

\begin{abstract}
Gallagher's ergodic theorem is a result in metric number theory. It states that the approximation
of real numbers by rational numbers obeys a striking {\lq}all or nothing{\rq} behaviour. We discuss
a formalisation of this result in the Lean theorem prover. As well as being notable in its own
right, the result is a key preliminary, required for Koukoulopoulos and Maynard's stunning recent
proof of the Duffin-Schaeffer conjecture.
\end{abstract}

\section{Introduction}
\label{sec:Introduction}
In addition to recognising extraordinary achievements of young mathematicians, the Fields Medal
provides a valuable service to the wider mathematical community: it draws attention to important
recent results. In recent years, such attention has had a significant positive impact in the
formalisation community. Buzzard, Commelin, and Massot's formalisation of the definition of a
perfectoid space \cite{BuzzardCommelinMassot20} and Commelin, Topaz et al.'s spectacular
success with the Liquid Tensor Experiment \cite{CommelinTopaz2022} were both
the results of projects which formalised work of 2018 Fields Medalist Peter Scholze. Amongst other
things, these projects demonstrate that today's proof assistants are capable of handling the
complicated constructions of contemporary mathematics.

In 2022, James Maynard was awarded a Fields Medal with a citation that highlighted his work on the
structure of prime numbers as well as on \emph{Diophantine approximation}. In the long form of the
citation we read that:
\begin{quote}
  Maynard has also produced fundamental work in Diophantine approximation, having solved
  the Duffin-Schaeffer conjecture with Koukoulopoulos.
\end{quote}
Shortly after the announcement of Maynard's award, Andrew Pollington suggested to the author that a
formalisation of Koukoulopoulos and Maynard's proof of the Duffin-Schaeffer conjecture would be a
worthy target for formalisation. Recognising that this would be an enormous undertaking, he
suggested focusing on various necessary preliminaries. Perhaps the most important of these is
Gallagher's ergodic theorem \cite{Gallagher1961} (see also \cite{KoukoulopoulosMaynard2020}
lemma 5.1).
The statement is as follows:
\begin{theorem}[Gallagher's theorem]\label{thm:gallagher_reals}
  Let $\delta_1, \delta_2, \ldots$ be a sequence of real numbers and let:
  \begin{align*}
    W = \{ x \in \R ~|~ \exists~ q \in \Q, |x - q| < \delta_{\denom(q)} ~i.o.\}.
  \end{align*}
  Then $W$ is almost equal to either $\emptyset$ or $\R$.
\end{theorem}
This deserves a few remarks:
\begin{itemize}
  \item The notation $\denom(q)$ denotes the denominator of $q$ (in lowest terms). It is a strictly
  positive natural number.
  \item Special attention should be paid to the letters {\lq}i.o.{\rq} appearing in the
  definition of $W$: these abbreviate the phrase {\lq}infinitely often{\rq}. The notation means
  that $x \in W$ iff there exists an \emph{infinite sequence} of rationals
  $q_0, q_1, \ldots$ (which may depend on $x$) with $\denom(q_0) < \denom(q_1) < \cdots$ satisfying
  $|x - q_i| < \delta_{\denom(q_i)}$ for all $i$.
  \item The phrase {\lq}almost equal{\rq} characterises this as a theorem of \emph{metric number
  theory}: it means that the sets are equal up to a set of Lebesgue measure zero.
  \item It is striking and not at all obvious that $W$ should exhibit such dichotomous behaviour.
\end{itemize}

It is the purpose of this article to discuss the author's formalisation of Gallagher's theorem.
It was carried out using the Lean proof assistant together with its \mathlib
library \cite{Mathlib}. More
precisely we formalised the following:
\begin{lstlisting}[caption={Gallagher's theorem \mllink{number_theory/well_approximable.lean\#L191}},label=lst:gallagher,captionpos=t]
theorem add_well_approximable_ae_empty_or_univ
  (δ : ℕ → ℝ) (hδ : tendsto δ at_top (𝓝 0)) :
  (∀ᵐ x, ¬ add_well_approximable 𝕊 δ x) ∨ ∀ᵐ x, add_well_approximable 𝕊 δ x :=
\end{lstlisting}
The notation will be explained in the pages to come. With further work, one could drop the
hypothesis $h\delta$, which says that $\delta_n \to 0$ as $n \to \infty$. In fact this highlights
a curious feature of the proof: one makes
with two totally separate arguments, a measure-theoretic argument which assumes the hypothesis
$h\delta$ and a number-theoretic argument assuming its negation. One then invokes
the law of excluded middle to deduce the result unconditionally. The measure-theoretic argument
assuming $h\delta$ is much harder and is what we have formalised.

The paper is intended for non-experts and the structure is as follows. In
section~\ref{sec:basic_concepts} we outline relevant
basic concepts so that we can reinterpret Gallagher's theorem as a result about the $\limsup$ of
thickenings of finite-order points in the circle. We also make some general remarks about
metric number theory. In section~\ref{sec:density_thm} we introduce the most important foundational
result required. Lebesgue's density theorem is the workhorse of Gallagher's proof. In
section~\ref{sec:cassels_lemma} we discuss a key measure-theoretic lemma due to Cassels which is of
some independent interest in its own right. In section~\ref{sec:ergodic_theory} we discuss the
results about ergodic maps which we needed, emphasising the ergodicity of certain maps of the
circle. In section~\ref{sec:gallagher_thm} we introduce points of approximately finite order and
use this language to give a proof of Gallagher's theorem. We finish with section~\ref{sec:final_words}
where, amongst other things, we discuss further directions this work could be
taken. For the most part we do not enter into the details of proofs. The main exception to this
is the proof of Gallagher's theorem itself since we hope our presentation of Gallagher's ideas
may make them more accessible than other accounts intended for specialists
(such as \cite{Gallagher1961} Theorem 1 or \cite{MR1672558} Theorem 2.7(B)).

In keeping with \mathlib's stress on mathematical unity, all work
was added directly to \mathlib's master branch in a series of 27 pull requests, collectively adding
just over 3,500 net new lines of code. This work is thus automatically
available to all future \mathlib users. Throughout this text we also provide permalinks to
relevant locations in \mathlib; each one is indicated with the symbol \link. We also provide a
judiciously chosen set of code listings (such as listing \ref{lst:gallagher} above) containing Lean
code. Often our intention is to assist the reader who wishes to compare a key informal statement
with its formal equivalent.

\section{Basic concepts}\label{sec:basic_concepts}
We outline some basic concepts to fix notation and to assist non-experts.
\subsection{Almost equal sets}
Given a measurable space with measure $\mu$, when there is no possibility of ambiguity about
the measure, we shall use the notation:
\begin{align}\label{def:ae_notation}
  s =_{a.e.} t,
\end{align}
to say two subsets $s$, $t$ are almost equal with respect to $\mu$. We recall that this is
equivalent to the following pair of measure-zero conditions
\mllink{measure_theory/measure/measure_space_def.lean\#L385}:
\begin{align}\label{ae_iff_measure_zero_pair}
  \mu(s \setminus t) = 0 \mbox{\quad and\quad } \mu(t \setminus s) = 0.
\end{align}

\subsection{Obeying a condition infinitely often}\label{subsect:inf_often}
Gallagher's theorem concerns a set of points obeying a condition {\lq}infinitely often{\rq}. In
general, given a sequence of subsets $s_0, s_1, \ldots$ of some background type $X$ the notation
$\exists ~\cdots~ i.o.$ is defined as\footnote{This is standard notation appearing throughout
the informal literature.}:
\begin{align}\label{def:inf_often}
  \{ x : X ~|~ \exists~ n \in \N, x \in s_n ~i.o.\} =
  \{ x : X ~|~ \mbox{the set $\{n \in \N ~|~ x \in s_n \}$ is infinite} \}.
\end{align}
In fact there is another expression for this set; it is easy to see that
\mllink{order/liminf_limsup.lean\#L860}:
\begin{align}\label{inf_often_eq_limsup}
  \{ x : X ~|~ \exists~ n \in \N, x \in s_n ~i.o.\} = \limsup s,
\end{align}
where:
\begin{align*}
  \limsup s = \bigcap_{n \ge 0} \bigcup_{i \ge n} s_i.
\end{align*}
When formalising a result about the set of points belonging to some family of subsets infinitely
often, one can thus phrase it
in language of \eqref{def:inf_often} or in the language of $\limsup$. We opted for the latter.
This was preferable because $\limsup$ makes sense for any complete lattice whereas
\eqref{def:inf_often} is specific to the lattice of subsets of a type. All API developed was
thus more widely applicable.

In the course of the proof it is useful to work with the $\limsup$ bounded by a predicate
$p : \N \to \Prop$. This can be defined:
\begin{align}\label{def:blimsup}
  \limsup_p s = \bigcap_{n \ge 0} \bigcup_{p(i), i \ge n} s_i.
\end{align}
The actual definition which we added
\texttt{filter.blimsup} \mllink{order/liminf_limsup.lean\#L285}
is slightly different so that it also applies in a conditionally complete lattice (such as
$\R$) but we provided a lemma showing the equivalence to \eqref{def:blimsup} for
complete lattices (such as \texttt{set $\R$}) \mllink{order/liminf_limsup.lean\#L578}. Using
\texttt{blimsup}, we can work with the $\limsup$ of the subfamily defined by the predicate $p$
without having to pass to the subtype of the indexing type. For example given two predicates $p$,
$q$, we can express the useful identity:
\begin{align}\label{blimsup_or}
  \limsup_{p \lor q} s = \limsup_p s \cup \limsup_q s,
\end{align}
formally as:
\begin{lstlisting}[caption={$\limsup$ bounded by the logical or of two predicates \mllink{order/liminf_limsup.lean\#L755}},label=lem:blimsup_or_eq_sup,captionpos=t]
@[simp] lemma blimsup_or_eq_sup :
  blimsup u f (λ x, p x ∨ q x) = blimsup u f p ⊔ blimsup u f q :=
\end{lstlisting}
without involving any subtypes. This is a standard design pattern used throughout \mathlib.

\subsection{Thickenings}
Using the language introduced above, the set $W$ appearing in the statement of
theorem \ref{thm:gallagher_reals} may be defined as:
\begin{align*}
  W = \limsup_{n > 0} s,
\end{align*}
where:
\begin{align*}
  s_n = \{ x \in \R ~|~ \exists~ q \in \Q, \denom(q) = n, |x - q| < \delta_n \}.
\end{align*}
These subsets $s_n$ have a special form: they are \emph{thickenings}. In general, given a metric
space $X$, if $s \subseteq X$ and $\delta \in \R$, the (open) $\delta$-thickening of $s$ is:
\begin{align*}
  \Th (\delta, s) = \{ x \in X ~|~ \exists~ y \in s, d(x, y) < \delta \}.
\end{align*}
This generalises the concept of an open ball. Fortunately thickenings already existed
\mllink{topology/metric_space/hausdorff_distance.lean\#L850} in \mathlib thanks to the work of
Gouezel on the Gromov-Hausdorff metric \cite{Gouezel2021}. We can thus express $W$ as:
\begin{align*}
  W = \limsup_{n > 0} \Th(\delta_n, \{ q \in \Q ~|~ \denom(q) = n \}).
\end{align*}
Using the language of thickenings turned out to be very convenient formally, not just for
Gallagher's theorem but also for example in lemma \ref{lem:cassels} (discussed below).

\subsection{The circle as a normed group}
The subset $W$ appearing in theorem \ref{thm:gallagher_reals} trivially satisfies the periodicity
condition\footnote{If $q$ approximates $x$ then $1 + q$ approximates $1 + x$ with the same
error and $\denom(1 + q) = \denom(q)$.}:
\begin{align*}
  1 + W = W,
\end{align*}
and thus descends to a subset of the circle $\Circ = \R / \Z$. This quotient is actually a
\emph{normed} group; given $x \in \R$ representing the coset $\hat x \in \Circ$, the norm obeys:
\begin{align}\label{circleNormLaw}
  \| \hat x\| = |x - \round (x)|
\end{align}
where $\round(x)$ is the nearest integer to $x$. Since the norm gives us a metric, we may
speak of thickenings of subsets of $\Circ$.

Furthermore, $\Q / \Z \subseteq \Circ$
is exactly the set of points of finite order in $\Circ$. Gallagher's theorem may thus be
regarded as establishing a special property enjoyed by the circle $\Circ$ in the category of normed
groups. As we shall see this is a useful point of view since the proof of Gallagher's theorem
depends on the fact that certain transformations are ergodic when regarded as maps
$\Circ \to \Circ$.

When this work began, \mathlib already contained a model of the circle as complex numbers of unit
length, called \texttt{circle} \mllink{analysis/complex/circle.lean\#L40}. Although this model is
naturally equivalent to $\R / \Z$, the equivalence uses the exponential and logarithm maps
which are irrelevant for our work. We thus introduced a second model called \texttt{add\_circle}
defined to be $\R / \Z$:
\begin{lstlisting}[caption={The additive circle \mllink{topology/instances/add_circle.lean\#L109}},label=def:add_circle,captionpos=t]
def add_circle {𝕜 : Type*} [linear_ordered_add_comm_group 𝕜]
  [topological_space 𝕜] [order_topology 𝕜] (p : 𝕜) :=
𝕜 ⧸ zmultiples p
\end{lstlisting}
When $p = 1$ this is exactly $\R / \Z$ but we allow a general value of $p$ to support other
applications\footnote{For example \mathlib uses $p = 2\pi$ to define angles
\mllink{analysis/special_functions/trigonometric/angle.lean\#L27}.}.

As the names suggest, \texttt{circle} carries an
instance of \mathlib's \texttt{group} class and \texttt{add\_circle} carries an instance of
\texttt{add\_group}. It is interesting that \mathlib's additive-multiplicative design pattern so
conveniently allows both models to coexist.

Substantial API for \texttt{add\_circle} was then developed, notably an
instance of the class \texttt{normed\_add\_comm\_group}
\mllink{analysis/normed/group/add_circle.lean\#L33} satisfying the identity \eqref{circleNormLaw}
\mllink{analysis/normed/group/add_circle.lean\#L262} and a characterisation of the finite-order
points using rational numbers \mllink{topology/instances/add_circle.lean\#L344}:
\begin{align}\label{circle_fin_order_rat}
  \{ y \in \Circ ~|~ o(y) = n \} = \{ [q] \in \Circ ~|~ q \in \Q, \denom(q) = n\},
\end{align}
where the notation $o(y) = n$ means that $y$ has order $n$.

Using all of our new language, the statement of Gallagher's theorem becomes:
\begin{theorem}[Gallagher's theorem]\label{thm:gallagher_circle}
  Let $\delta_1, \delta_2, \ldots$ be a sequence of real numbers and let:
  \begin{align*}
    W = \limsup_{n > 0} \Th(\delta_n, \{ y \in \Circ ~|~ o(y) = n \}).
  \end{align*}
  Then $W =_{a.e.} \emptyset$ or $W =_{a.e.} \Circ$.
\end{theorem}
Using \eqref{circleNormLaw} and \eqref{circle_fin_order_rat}, theorem
\ref{thm:gallagher_circle} is trivially equivalent to theorem \ref{thm:gallagher_reals}.

\subsection{Metric number theory}
Metric number theory is the study of arithmetic properties of the real numbers (and related spaces)
which hold {\lq}almost everywhere{\rq} with respect to the Lebesgue measure. The arithmetic property
in the case of Gallagher's theorem is approximation by rational numbers.

To illustrate, consider the set $\mathbb{I}$ of real numbers which have infinitely-many
quadratically-close rational approximations:
\begin{align*}
  \mathbb{I} &= \{ x \in \R ~|~ \exists~ q \in \Q, |x - q| < 1 / \denom(q)^2 ~i.o. \}\\
             &= \limsup_{n > 0} \Th(1/n^2, \{ q \in \Q ~|~ \denom(q) = n \}).
\end{align*}
It has been known at least since the early 19$^{\rm th}$ Century, that $\mathbb{I}$ is just the set
of irrational numbers\footnote{Thanks to Michael Geißer and Michael Stoll, \mathlib knows this fact
\mllink{number_theory/diophantine_approximation.lean\#L258}.}:
\begin{align}\label{I_eq_R_less_Q}
  \mathbb{I} = \R \setminus \Q.
\end{align}

Considering this result from the point of view of metric number theory, we notice that since
$\Q$ has Lebesgue measure zero, $\mathbb{I}$ is almost equal to $\R$. Thus the metric number
theorist would be content to summarise \eqref{I_eq_R_less_Q} by saying that
\begin{align*}
  \mathbb{I} =_{a.e.} \R,
\end{align*}
without worrying about exactly which numbers $\mathbb{I}$ contains.

The benefit of metric number theorist's point of view is that a great many questions have answers of
this shape. Gallagher's theorem is an especially-beautiful example of this phenomenon.

\section{Doubling measures and Lebesgue's density theorem}\label{sec:density_thm}
Lebesgue's density theorem is a foundational result in measure theory, required for the proof of
Gallagher's theorem. Although we only
needed to apply it to the circle, the density theorem holds quite generally and so we took some
trouble to formalise it subject to quite weak assumptions\footnote{We were lucky that Sébastien
Gouëzel had recently added an extremely general theory of Vitali families which made this
possible.}.

\subsection{Doubling measures}
A convenient class of measures for which the density theorem holds is the class of doubling
measures.
\begin{definition}
  Let $X$ be a measurable metric space carrying a measure $\mu$. We say $\mu$ is a doubling measure
  if there exists $C \ge 0$ and $\delta > 0$ such that for all $0 < \epsilon \le \delta$ and
  $x \in X$:
  \begin{align*}
    \mu (B(x, 2\epsilon)) \le C\mu (B(x, \epsilon)).
  \end{align*}
  where $B(x, r)$ denotes the closed ball of radius $r$ about $x$.
\end{definition}
The corresponding formal definition, which the author added to \mathlib for the purposes of
formalising the density theorem, is:
\begin{lstlisting}[caption={Definition of doubling measures \mllink{measure_theory/measure/doubling.lean\#L39}},label=def_doubling_measure,captionpos=t]
class is_doubling_measure
  {α : Type*} [metric_space α] [measurable_space α] (μ : measure α) :=
  (exists_measure_closed_ball_le_mul [] : ∃ (C : ℝ≥0), ∀ᶠ ε in 𝓝[>] 0, ∀ x,
    μ (closed_ball x (2 * ε)) ≤ C * μ (closed_ball x ε))
\end{lstlisting}
The parameter $\delta$ is not explicitly mentioned in the code above because we use \mathlib's
standard notation for the concept of a predicate holding eventually along a filter
\mllink{order/filter/basic.lean\#L953}.

For our application, we needed to apply the density theorem to the Haar measure \cite{Door21}
\mllink{measure_theory/integral/periodic.lean\#L58} on the circle. Of course this turns out to be
the familiar arc-length measure and so the volume of a closed ball of radius
$\epsilon$ is given by \mllink{measure_theory/integral/periodic.lean\#L88}:
\begin{align*}
  \mu(B(x, \epsilon)) = \min(1, 2\epsilon).
\end{align*}
Taking $C = 2$ we thus see that the Haar measure on the circle is
doubling. We registered this fact using a typeclass instance as follows:
\begin{lstlisting}[caption={The circle's doubling measure \mllink{measure_theory/integral/periodic.lean\#L110}},label=def_circle_doubling,captionpos=t]
instance : is_doubling_measure (volume : measure (add_circle T)) :=
\end{lstlisting}
The unit circle corresponds to taking $T = 1$, but the code allows any $T > 0$. Thanks to
this instance, Lean knows that any results proved for doubling measures automatically holds for the
Haar measure on the circle.

\subsection{The density theorem}
The version of the density theorem which we formalised is:
\begin{theorem}\label{density_thm}
  Let $X$ be a measurable metric space carrying a measure $\mu$. Suppose that $X$ has
  second-countable topology and that $\mu$ is doubling and locally finite. Let $S \subseteq X$ and
  $K \in \R$, then for almost all $x \in S$, given any sequence of points $w_0, w_1, \ldots$
  and distances $\delta_0, \delta_1, \ldots$, if:
  \begin{itemize}
    \item $\delta_j \to 0$ as $j \to \infty$ and,
    \item $x \in B(w_j, K\delta_j)$ for large enough $j$,
  \end{itemize}
  then:
  \begin{align*}
    \frac{\mu (S \cap B(w_j, \delta_j))}{\mu (B(w_j, \delta_j))} \to 1,
  \end{align*}
  as $j \to \infty$.
\end{theorem}
Even in the special case $K = 1$ and $w_0 = w_1 = \cdots = x$, the result is quite
powerful\footnote{Indeed this is probably the most common version one finds in the literature.}.
A point $x$ satisfying the property appearing in the theorem statement is known as a point of
density 1. Using this language, Lebesgue's density theorem asserts that almost all points of a
set have density 1. In particular if $\mu(S) > 0$ then there must exist a point of density 1
\mllink{measure_theory/measure/measure_space.lean\#L1618}.
As an example, if $X = \R$ and $S$ is the closed interval $[0, 1]$, the set of points of density 1
is the open interval $(0, 1)$.

In fact the formal version which we added to \mathlib is very slightly more general since it allows
$w$ and $\delta$ to be maps from any space carrying a filter. After a preparatory \texttt{variables}
statement:
\begin{lstlisting}[caption={Variables for the density theorem},label=density_variables,captionpos=t]
variables {α : Type*} [metric_space α] [measurable_space α] (μ : measure α)
  [is_doubling_measure μ] [second_countable_topology α] [borel_space α]
  [is_locally_finite_measure μ]
\end{lstlisting}
it looks like this:
\begin{lstlisting}[caption={Lebesgue's density theorem for doubling measures \mllink{measure_theory/covering/density_theorem.lean\#L143}},label=lem_density_thm,captionpos=t]
lemma is_doubling_measure.ae_tendsto_measure_inter_div (S : set α) (K : ℝ) :
  ∀ᵐ x ∂μ.restrict S, ∀ {ι : Type*} {l : filter ι} (w : ι → α) (δ : ι → ℝ)
    (δlim : tendsto δ l (𝓝[>] 0))
    (xmem : ∀ᶠ j in l, x ∈ closed_ball (w j) (K * δ j)), tendsto (λ j,
    μ (S ∩ closed_ball (w j) (δ j)) / μ (closed_ball (w j) (δ j))) l (𝓝 1) :=
\end{lstlisting}
The method of proof is essentially is to develop sufficient API for \texttt{is\_doubling\_measure}
to show that such measure spaces carry certain natural families of subsets called Vitali families
and then to invoke the lemma \texttt{vitali\_family.ae\_tendsto\_measure\_inter\_div}
\mllink{measure_theory/covering/differentiation.lean\#L750} added by Gouëzel as part of an
independent project \cite{Gouezel22}.

\section{Cassels's lemma}\label{sec:cassels_lemma}
A key ingredient in the proof of Gallagher's theorem is the following result due to Cassels.
\begin{lemma}\label{lem:cassels}
  Let $X$ be a measurable metric space carrying a measure $\mu$. Suppose that $X$ has
  second-countable topology and that $\mu$ is doubling and locally finite. Let
  $s_0, s_1, \ldots$ be a sequence of subsets of $X$ and $r_0, r_1, \ldots$ be a sequence of real
  numbers such that $r_n \to 0$ as $n \to \infty$. For any $M > 0$ let:
  \begin{align*}
    W_M = \limsup \Th (Mr_n, s_n),
  \end{align*}
  then:
  \begin{align*}
    W_M =_{a.e.} W_1,
\end{align*}
  i.e., up to sets of measure zero, $W_M$ does not depend on $M$.
\end{lemma}
This essentially appears as lemma 9 in \cite{cassels1950} in the special case that:
\begin{enumerate}[(a)]
  \item $X$ is the open interval $(0, 1)$,
  \item $\mu$ is the Lebesgue measure,
  \item $s_n$ is a sequence of points rather than a sequence of subsets.\label{generalise_subsets}
\end{enumerate}

Reusing the \texttt{variables} from listing \ref{density_variables}, the formal version of
lemma \ref{lem:cassels} which we added to \mathlib looks like this:
\begin{lstlisting}[caption={Cassels's lemma \mllink{measure_theory/covering/liminf_limsup.lean\#L268}},label=lem_cassels,captionpos=t]
theorem blimsup_thickening_mul_ae_eq
  (p : ℕ → Prop) (s : ℕ → set α) {M : ℝ} (hM : 0 < M)
  (r : ℕ → ℝ) (hr : tendsto r at_top (𝓝 0)) :
  (blimsup (λ i, thickening (M * r i) (s i)) at_top p : set α) =ᵐ[μ]
  (blimsup (λ i, thickening (r i) (s i)) at_top p : set α) :=
\end{lstlisting}
Several remarks are in order:
\begin{itemize}
  \item The syntax \texttt{s =$^{\rm m}$[$\mu$] t} is \mathlib's notation for sets (or functions)
  $s$, $t$ being almost equal with respect to a measure $\mu$. It is the formal equivalent of the
  popular informal notation \eqref{def:ae_notation}.
  \item The type ascriptions \texttt{: set $\alpha$} appear because of an unresolved
  \emph{typeclass diamond} in \mathlib's library of lattice theory. The issue
  is that the type \texttt{set $\alpha$} is definitionally equal to $\alpha \to \Prop$. Since
  $\Prop$ is a complete boolean algebra \mllink{order/complete_boolean_algebra.lean\#L226}
  it follows \mllink{order/complete_boolean_algebra.lean\#L222} that $\alpha \to \Prop$ is a
  complete boolean algebra. Unfortunately the definition \mllink{data/set/lattice.lean\#L121}
  of the complete boolean algebra structure on \texttt{set $\alpha$}, though mathematically equal,
  is not definitionally equal to that on $\alpha \to \Prop$. Strictly speaking, because
  \texttt{set $\alpha$} is a type synonym, this is a permissible diamond but it would still be
  useful to resolve it.\footnote{The diamond is recorded in \mathlib issue 16932
  \href{https://github.com/leanprover-community/mathlib/issues/16932}{\link}.
  In fact it is only the \texttt{Inf} and \texttt{Sup} fields in the complete boolean algebra
  structures that differ definitionally so this should be fairly easy to resolve.}
  \item Listing \ref{lem_cassels} is stated in terms of \texttt{blimsup}, i.e., a $\limsup$
  bounded by a predicate $p$. As discussed in section \eqref{subsect:inf_often}, this allows us to
  avoid having to deal with subtypes. We will see that this is convenient when applying this lemma
  in the proof of Gallagher's theorem.
  \item The key ingredient in the proof of Cassels's lemma is Lebesgue's density theorem
  \ref{density_thm}. In view of \eqref{ae_iff_measure_zero_pair}, Cassels's lemma requires us to
  establish a pair of measure-zero conditions.
  According to whether $M < 1$ or $M > 1$, exactly one of these two conditions is
  trivial for the two sets appearing in the statement of Cassels's lemmma \ref{lem:cassels}. To
  prove the non-trivial measure-zero condition, one
  argues by contradiction by assuming the measure is strictly positive, applying the density theorem
  to obtain a point of density 1, and showing that this is impossible for a doubling measure. The
  only non-trivial dependency is Lebesgue's density theorem.
  \item Although the modifications required for the generalisation of this lemma from its original
  form in \cite{cassels1950} are straightforward, the generalisation \eqref{generalise_subsets}
  from points to subsets (equivalently from balls to thickenings) is extremely useful formally. In
  the application of this lemma required
  for Gallagher's theorem, $s_n$ is the set of points of order $n$ in the circle. In the informal
  literature, the version of lemma \ref{lem:cassels} for sequences of points can be applied because
  the circle has only finitely-many points of each finite order and so one can enumerate all points
  of finite order as a single sequence of points. This would be messy formally.
\end{itemize}

\section{Ergodic theory}\label{sec:ergodic_theory}
Ergodic theory is the study of measure-preserving maps. Given measure spaces
$(X, \mu_X)$ and $(Y, \mu_Y)$, a measurable map $f : X \to Y$ is measure-preserving if:
\begin{align*}
  \mu_X(f^{-1}(s)) = \mu_Y(s),
\end{align*}
for any measurable set $s \subseteq Y$. For example, given any $c \in \R$, taking Lebesgue
measure on both domain and codomain, the translation $x \mapsto c + x$ is always measure-preserving
whereas the dilation $x \mapsto cx$ is measure-preserving only if $c = \pm 1$. Fortunately
\mathlib already contained an excellent theory of measure-preserving maps.

\subsection{Ergodic maps, general theory}
Within ergodic theory, special attention is paid to ergodic maps.
\begin{definition}
  Let $(X, \mu)$ be a measure space and $f : X \to X$ be measure-preserving.
  We say $f$ is ergodic if for any measurable set $s \subseteq X$:
  \begin{align*}
    f^{-1}(s) = s \implies \mbox{$s$ is almost equal to $\emptyset$ or $X$}.
  \end{align*}
\end{definition}
Ergodicity is key concept in the proof of Gallagher's theorem and so we added the following
definitions to \mathlib:
\begin{lstlisting}[caption={Definition of pre-ergodic \mllink{dynamics/ergodic/ergodic.lean\#L39} and ergodic maps \mllink{dynamics/ergodic/ergodic.lean\#L44}:},label=def:ergodic,captionpos=t]
structure pre_ergodic (μ : measure α . volume_tac) : Prop :=
  (ae_empty_or_univ : ∀ ⦃s⦄, measurable_set s →
    f⁻¹' s = s → s =ᵐ[μ] (∅ : set α) ∨ s =ᵐ[μ] univ)

structure ergodic (μ : measure α . volume_tac) extends
  measure_preserving f μ μ, pre_ergodic f μ : Prop
\end{lstlisting}
The reason for the intermediate definition \texttt{pre\_ergodic} is to support the definition of
quasi-ergodic maps which we also defined, but which do not concern us here.

We then developed some basic API for ergodic maps including the key result:
\begin{lemma}\label{lemma:ae_empty_or_univ_of_image_ae_le}
  Let $X$ be a measurable space with measure $\mu$ such that $\mu(X) < \infty$. Suppose that
  $f : X \to X$ is ergodic, $s \subseteq X$ is measurable, and the image $f(s)$ is almost contained
  in $s$, then $s$ is almost equal to $\emptyset$ or $X$.
\end{lemma}
This result is elementary but not quite trivial and appears formally as follows:
\begin{lstlisting}[caption={Sets that are almost invariant by an ergodic map \mllink{dynamics/ergodic/ergodic.lean\#L179}:},label=lem:ae_empty_or_univ_of_image_ae_le,captionpos=t]
lemma ae_empty_or_univ_of_image_ae_le [is_finite_measure μ]
  (hf : ergodic f μ) (hs : measurable_set s) (hs' : f '' s ≤ᵐ[μ] s) :
  s =ᵐ[μ] (∅ : set X) ∨ s =ᵐ[μ] univ :=
\end{lstlisting}
This is not the first time that ergodic maps have been formalised in a theorem prover and so we have
kept the above account very brief. Indeed the Archive of Formal Proofs for Isabelle/HOL contains an
impressive body of results about ergodic theory due to Sébastien Gouëzel with contributions from
Manuel Eberl, available at the Ergodic Theory entry
\href{https://www.isa-afp.org/entries/Ergodic_Theory.html}{\link}. This entry contains many
results about general ergodic theory that have not yet been added to \mathlib. On the other
hand, we needed to know that certain specific maps on the circle are ergodic and our formalisations
of these results do appear to be the first of their kind. We discuss these next.

\subsection{Ergodic maps on the circle}
In order to prove Gallagher's theorem, we needed the following result:
\begin{theorem}\label{thm:circle_smul_ergodic}
  Given $n \in \N$, the map:
  \begin{align*}
    \Circ &\to \Circ\\
    y &\mapsto ny
  \end{align*}
  is measure-preserving if $n \ge 1$ and is ergodic if $n \ge 2$.
\end{theorem}
The fact that $y \mapsto ny$ is measure-preserving follows from general uniqueness results for
Haar measures. In fact the result holds for any compact, Abelian, divisible topological group.
Thanks to \mathlib's extensive theory of Haar measure \cite{Door21}, it was easy to add a proof of
this \mllink{measure_theory/measure/haar.lean\#L741}. We encourage readers who are
encountering this fact for the first time to examine figure \ref{fig:measure_preserving} and
appreciate why this result holds for $\Circ$ despite failing for $\R$.
\begin{figure}
  \centering
  \begin{tikzpicture}
    \node [circle,minimum size=80pt,draw] at (0pt,0pt){};
    \draw [ultra thick] (40pt,0pt) arc [start angle=0, end angle=90, x radius=40pt, y radius=40pt];
    \node [circle,minimum size=80pt,draw] at (200pt,0pt){};
    \draw [ultra thick] (240pt,0pt) arc [start angle=0, end angle=45, x radius=40pt, y radius=40pt];
    \draw [ultra thick] (160pt,0pt) arc [start angle=180, end angle=225, x radius=40pt, y radius=40pt];
    \node at (0pt, 60pt) {A subset $s \subseteq \Circ$};
    \node at (200pt, 60pt) {Its preimage $f^{-1}(s)$};
  \end{tikzpicture}
  \caption{The map $f : y \mapsto 2y$ is measure-preserving.}
  \label{fig:measure_preserving}
\end{figure}
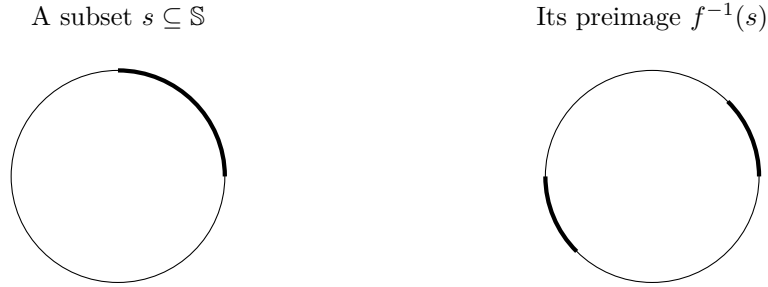

The proof that $y \mapsto ny$ is ergodic is harder. We proved it as corollary of the following
lemma. We sketch a proof to give a sense of what is involved; it is not essential
that the reader follow the details: the main point is that we needed to use Lebesgue's density
theorem.
\begin{lemma}\label{lem:circle_seq_ergodic}
  Let $s \subseteq \Circ$ be measurable and $u_0, u_1, \ldots$ be a sequence of finite-order
  points in $\Circ$ such that:
  \begin{itemize}
    \item $u_i + s$ is almost equal to $s$ for all $i$,
    \item the order $o(u_i) \to \infty$ as $i \to \infty$.
  \end{itemize}
  Then $s$ is almost equal to $\emptyset$ or $X$.
\end{lemma}
\begin{proof}
The result is fairly intuitive: $s$ is almost equal to
$u_i + s$ iff it is composed of a collection of $o(u_i)$ components, evenly-spaced throughout the
circle, up to a set of measure zero. Since this holds for all $i$ and $o(u_i) \to \infty$,
such components must either fill out the circle or be entirely absent, up to a set of measure zero.

The way to turn the above intuitive argument into rigorous proof is to use Lebesgue's density
theorem \ref{density_thm}. We must show that if $s$ is not almost empty then
$\mu(s) = 1$. Lebesgue tells us that if $s$ is not almost empty it must contain some point $d$ of
density 1. Using $d$, we construct the sequence of closed balls $B_i$ centred on $d$ such that
$\mu(B_i) = o(u_i)$. Because $u_i + s$ is almost $s$,
\begin{align*}
  \mu(s \cap B_i) = o(u_i) \mu(s) = \mu(B_i)\mu(s).
\end{align*}
However since $d$ has density 1, we know that:
\begin{align*}
  \mu(s \cap B_i) / \mu(B_i) \to 1.
\end{align*}
These two results force us to conclude that $\mu(s) = 1$.
\end{proof}

The formal version is very slightly more general and appears in \mathlib as follows:
\begin{lstlisting}[caption={Formal statement of lemma \ref{lem:circle_seq_ergodic} \mllink{dynamics/ergodic/add_circle.lean\#L40}},label=lst:circle_seq_ergodic,captionpos=t]
lemma add_circle.ae_empty_or_univ_of_forall_vadd_ae_eq_self
  {s : set $ add_circle T} (hs : null_measurable_set s volume)
  {ι : Type*} {l : filter ι} [l.ne_bot] {u : ι → add_circle T}
  (hu₁ : ∀ i, ((u i) +ᵥ s : set _) =ᵐ[volume] s)
  (hu₂ : tendsto (add_order_of ∘ u) l at_top) :
  s =ᵐ[volume] (∅ : set $ add_circle T) ∨ s =ᵐ[volume] univ :=
\end{lstlisting}

Theorem \ref{thm:circle_smul_ergodic} follows from lemma \ref{lem:circle_seq_ergodic} because
any set $s$ satisfying $f^{-1}(s) = s$ for $f : y \mapsto ny$ satisfies $u_i + s = s$ for the
sequence:
\begin{align*}
  u_i = [1 / n^i] \in \Circ.
\end{align*}
Note that we need $n \ge 2$ in order to have $o(u_i) = n^i \to \infty$.
The formal statement appears in \mathlib as follows:
\begin{lstlisting}[caption={Formal statement of theorem \ref{thm:circle_smul_ergodic} \mllink{dynamics/ergodic/add_circle.lean\#L133}},label=lst:circle_smul_ergodic,captionpos=t]
lemma ergodic_nsmul {n : ℕ} (hn : 1 < n) :
  ergodic (λ (y : add_circle T), n • y) :=
\end{lstlisting}

In fact we needed the following mild generalisation of theorem \ref{thm:circle_smul_ergodic}:
\begin{theorem}\label{thm:circle_smul_add_ergodic}
  Given $n \in \N$ and $x \in \Circ$, the map:
  \begin{align*}
    \Circ &\to \Circ\\
    y &\mapsto ny + x
  \end{align*}
  is measure-preserving if $n \ge 1$ and is ergodic if $n \ge 2$.
\end{theorem}
This follows easily from theorem \ref{thm:circle_smul_ergodic} because if we define the
measure-preserving equivalence:
\begin{align*}
  e : \Circ &\to \Circ\\
  y &\mapsto \frac{x}{n-1} + y
\end{align*}
then a quick calculation reveals:
\begin{align*}
  e \circ g \circ e^{-1} = f,
\end{align*}
where $f : y \mapsto ny$ and $g : y \mapsto ny + x$. As a result, theorem
\ref{thm:circle_smul_add_ergodic} follows from theorem \ref{thm:circle_smul_ergodic} via:
\begin{lstlisting}[caption={The reduction of theorem \ref{thm:circle_smul_add_ergodic} to theorem \ref{thm:circle_smul_ergodic} \mllink{dynamics/ergodic/ergodic.lean\#L99}},label=lst:ergodic_conjugate_iff,captionpos=t]
lemma ergodic_conjugate_iff {e : α ≃ᵐ β} (h : measure_preserving e μ μ') :
  ergodic (e ∘ f ∘ e.symm) μ' ↔ ergodic f μ :=
\end{lstlisting}

\section{Gallagher's theorem}\label{sec:gallagher_thm}
\subsection{Points of approximate order}
Recall the definition of the set $W \subseteq \Circ$ appearing in the statement of theorem
\ref{thm:gallagher_circle}:
\begin{align*}
  W = \limsup_{n > 0} \Th(\delta_n, \{ y \in \Circ ~|~ o(y) = n \}).
\end{align*}
Key to the proof of theorem \ref{thm:gallagher_circle} is the way in which the sets
$\Th(\delta_n, \{ y \in \Circ ~|~ o(y) = n \})$ interact with the group structure of $\Circ$.
We thus made the following definition:
\begin{definition}\label{def:approx_order}
  Let $A$ be a seminormed group, $n \in \N$ (non-zero), and $\delta \in \R$. We shall use the
  notation:
  \begin{align*}
    \mathbb{AO}(A, n, \delta) = \Th(\delta, \{ y \in A ~|~ o(y) = n \}),
  \end{align*}
  for the set of points that have approximate order $n$, up to a distance $\delta$.
\end{definition}
For example, as shown in figure \ref{fig:approx_order}, $\mathbb{AO}(\Circ, n, \delta)$ is a union
of $\varphi(n)$ arcs of diameter $2\delta$, centred on the points $[m/n]$ with $0 \le m < n$ and $m$
coprime to $n$ (where $\varphi$ is Euler's totient function).
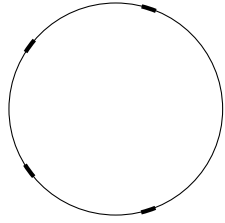
\begin{figure}
  \centering
  \begin{tikzpicture}
    \node [circle,minimum size=80pt,draw] at (0pt,0pt){};
    \draw [ultra thick] (15pt,37pt) arc [start angle=68, end angle=76, x radius=40pt, y radius=40pt];
    \draw [ultra thick] (-30.5pt,26pt) arc [start angle=140, end angle=148, x radius=40pt, y radius=40pt];
    \draw [ultra thick] (-34pt,-21pt) arc [start angle=212, end angle=220, x radius=40pt, y radius=40pt];
    \draw [ultra thick] (9.5pt,-39pt) arc [start angle=284, end angle=292, x radius=40pt, y radius=40pt];
  \end{tikzpicture}
  \caption{The set of points of approximate order 5 in $\Circ$, up to a distance
  $\delta \approx 0.01$.}
  \label{fig:approx_order}
\end{figure}

The formal counterpart of definition \ref{def:approx_order} is:
\begin{lstlisting}[caption={Points of approximate order in a normed group \mllink{number_theory/well_approximable.lean\#L62}},label=lst:approx_add_order_of,captionpos=t]
@[to_additive] def approx_order_of
  (A : Type*) [seminormed_group A] (n : ℕ) (δ : ℝ) : set A :=
thickening δ {y | order_of y = n}
\end{lstlisting}
Using this language, the only properties of $\mathbb{AO}(A, n, \delta)$ that we needed are
as follows:
\begin{lemma}\label{lem:AF_API}
  Let $A$ be a seminormed commutative group, $\delta \in \R$, $a \in A$, and $m, n \in \N$
  (both non-zero). Then\footnote{If $s \subseteq A$ the notation $m \cdot s$ means
  $\{ my ~|~ y \in s \}$.}:
  \begin{enumerate}[(i)]
    \item\label{AF_part_i} $m \cdot \mathbb{AO}(A, n, \delta) \subseteq \mathbb{AO}(A, n, m\delta)$
    if $m$, $n$ are coprime \mllink{number_theory/well_approximable.lean\#L90},
    \item\label{AF_part_ii} $m \cdot \mathbb{AO}(A, nm, \delta) \subseteq \mathbb{AO}(A, n, m\delta)$
    \mllink{number_theory/well_approximable.lean\#L102},
    \item\label{AF_part_iii} $a + \mathbb{AO}(A, n, \delta) \subseteq \mathbb{AO}(A, o(a)n, \delta)$
    if $o(a)$ and $n$ are coprime \mllink{number_theory/well_approximable.lean\#L115},
    \item\label{AF_part_iv} $a + \mathbb{AO}(A, n, \delta) = \mathbb{AO}(A, n, \delta)$ if
    $o(a)^2$ divides $n$ \mllink{number_theory/well_approximable.lean\#L128}.
  \end{enumerate}
\end{lemma}

In fact property \eqref{AF_part_iv} holds under the weaker assumption that $r(o(a))o(a)$
divides $n$ where $r(l)$ denotes the radical of a natural number $l$, but we needed only the version
stated in the lemma.

We made one last definition in support of theorem \ref{thm:gallagher_circle}:
\begin{definition}\label{def:well_approx}
  Let $A$ be a seminormed group and $\delta_1, \delta_2, \ldots$ a sequence of real numbers.
  We shall use the notation:
  \begin{align*}
    \mathbb{WA}(A, \delta) = \limsup_{n > 0} \mathbb{AO}(A, n, \delta_n),
  \end{align*}
  for the set of elements of $A$ that are well-approximable by points of finite order, relative to
  $\delta$.
\end{definition}
Note that $W = \mathbb{WA}(\Circ, \delta)$ where $W$ is the set appearing in the statement
of theorem \ref{thm:gallagher_circle}. The formal counterpart of definition \ref{def:well_approx}
is:
\begin{lstlisting}[caption={The set of well-approximable elements of a normed group \mllink{number_theory/well_approximable.lean\#L77}},label=lst:def_well_approx,captionpos=t]
@[to_additive] def well_approximable
  (A : Type*) [seminormed_group A] (δ : ℕ → ℝ) : set A :=
blimsup (λ n, approx_order_of A n (δ n)) at_top (λ n, 0 < n)
\end{lstlisting}
The additive version of this definition is \texttt{add\_well\_approximable}.

\subsection{The main theorem}
We are finally in a position to assemble everything and provide a proof of our main result.
For the reader's convenience we reproduce the formal statement which appeared above
in listing \ref{lst:gallagher}:
\begin{lstlisting}[caption={Gallagher's theorem \mllink{number_theory/well_approximable.lean\#L191}},label=lst:gallagher2,captionpos=t]
theorem add_well_approximable_ae_empty_or_univ
  (δ : ℕ → ℝ) (hδ : tendsto δ at_top (𝓝 0)) :
  (∀ᵐ x, ¬ add_well_approximable 𝕊 δ x) ∨ ∀ᵐ x, add_well_approximable 𝕊 δ x :=
\end{lstlisting}
The notation $\forall^{m} x, \cdots$ should be read {\lq}for almost all $x \cdots${\rq} and is
standard \mathlib notation \mllink{measure_theory/measure/measure_space_def.lean\#L512}.
Using the lemmas \texttt{filter.eventually\_eq\_empty}
\mllink{order/filter/basic.lean\#L1375} and \texttt{filter.eventually\_eq\_univ}
\mllink{order/filter/basic.lean\#L1272} the statement in listing \ref{lst:gallagher2} is equivalent
to:
\begin{theorem}[Gallagher's theorem with $\delta \to 0$]\label{thm:gallagher_circle_decaying}
  Let $\delta_1, \delta_2, \ldots$ be a sequence of real numbers such that $\delta_n \to 0$ as
  $n \to \infty$. Then $\mathbb{WA}(\Circ, \delta)$ is almost equal to either $\emptyset$ or
  $\Circ$.
\end{theorem}
\begin{proof}
  For each prime $p \in \N$ we define three sets\footnote{The notation $p \Vert n$ means that $p$
  divides $n$ exactly once.}:
  \begin{align*}
    A_p &= \limsup_{n > 0, p \nmid n} \mathbb{AO}(\Circ, n, \delta_n),\\
    B_p &= \limsup_{n > 0, p \Vert n} \mathbb{AO}(\Circ, n, \delta_n),\\
    C_p &= \limsup_{n > 0, p^2 \mid n} \mathbb{AO}(\Circ, n, \delta_n).
  \end{align*}
  Let $W = \mathbb{WA}(\Circ, \delta)$; bearing in mind \eqref{blimsup_or} it is clear that for any
  $p$:
  \begin{align}\label{W_eq_ABC}
    W = A_p \cup B_p \cup C_p.
  \end{align}
  We claim that these sets have the following properties:
  \begin{enumerate}[(a)]
    \item\label{Ap_claim} $A_p$ is almost invariant under the ergodic map: $y \mapsto py$,
    \item\label{Bp_claim} $B_p$ is almost invariant under the ergodic map: $y \mapsto py + [1/p]$,
    \item\label{Cp_claim} $C_p$ is invariant under the map $y \mapsto y + [1/p]$.
  \end{enumerate}

  To see why \eqref{Ap_claim} holds, consider:
  \begin{align*}
    p \cdot A_p &= p \cdot \limsup_{n > 0, p \nmid n} \mathbb{AO}(\Circ, n, \delta_n)\\
                &\subseteq \limsup_{n > 0, p \nmid n} p \cdot \mathbb{AO}(\Circ, n, \delta_n)\\
                &\subseteq \limsup_{n > 0, p \nmid n} \mathbb{AO}(\Circ, n, p \delta_n) &&\mbox{by lemma \ref{lem:AF_API} part \eqref{AF_part_i}}\\
                &=_{a.e.} A_p &&\mbox{by lemma \ref{lem:cassels}.}
  \end{align*}

  A very similar argument shows why \eqref{Bp_claim} holds except using parts \eqref{AF_part_ii},
  \eqref{AF_part_iii} of lemma \ref{lem:AF_API} instead of part \eqref{AF_part_i}.

  Claim \eqref{Cp_claim} is actually the most straightforward and holds by direct application of
  lemma \ref{lem:AF_API} part \eqref{AF_part_iv}.

  Now if $A_p$ is not almost empty for any prime $p$, then because it is almost invariant under an
  ergodic map, lemma \ref{lemma:ae_empty_or_univ_of_image_ae_le} tells us that it must be almost
  equal to $\Circ$. Since $A_p \subseteq W$, $W$ must also almost equal $\Circ$ and we have nothing
  left to prove.

  We may thus assume $A_p$ is almost empty for all primes $p$. By an identical argument, we may also
  assume $B_p$ is almost empty for all primes $p$. In view of \eqref{W_eq_ABC}, this means that:
  \begin{align*}
    W =_{a.e.} C_p \mbox{\quad for all $p$}.
  \end{align*}
  Thus, by \eqref{Cp_claim}, $W$ is almost
  invariant under the map $y \mapsto y + [1/p]$ for all primes $p$. The result then follows by
  applying lemma \ref{lem:circle_seq_ergodic}.
\end{proof}
Omitting code comments, the formal version of this $\sim 30$ line informal proof in \mathlib
requires 101 lines \mllink{number_theory/well_approximable.lean\#L191}.

\section{Final words}\label{sec:final_words}
\subsection{Removing the $\delta_n \to 0$ hypothesis}
As mentioned in the introduction, the hypothesis that $\delta_n \to 0$ in theorem
\ref{thm:gallagher_circle_decaying} may be removed.
A nice follow-up project would be to supply the proof in this case. By replacing $\delta_n$ with
$\max(\delta_n, 0)$, we may assume $0 \le \delta_n$ for all $n$. Given this, if
$\delta_n \not \to 0$, then in fact:
\begin{align*}
  \mathbb{WA}(\Circ, \delta) = \Circ.
\end{align*}
Note that this is a true equality of sets; it is not a measure-theoretic result. The main effort
would be to establish some classical bounds on the growth of the divisor-count and totient
functions.

In fact Bloom and Mehta have already formalised some of the required bounds as part of
their impressive Unit Fractions Project
\href{https://github.com/b-mehta/unit-fractions/blob/b60c39a3ebd40a84104a4064840b10fc2af15fb8/src/for_mathlib/basic_estimates.lean\#L824}{\link}
formalising Bloom's breakthrough \cite{Bloom2021, BloomMehta2022}. Once the relevant results are
migrated to \mathlib, removing the $\delta_n \to 0$ hypothesis will become even easier.

\subsection{The Duffin-Schaeffer conjecture}
Given some sequence of real numbers $\delta_1, \delta_2, \ldots$, Gallagher's theorem tells us that
$\mathbb{WA}(\Circ, \delta)$ is almost equal to either $\emptyset$ or to $\Circ$. The obvious
question is how to tell which of these two possibilities actually occurs for the sequence in hand.
The Duffin-Schaeffer conjecture, now a theorem thanks to Koukoulopoulos and Maynard, provides a
very satisfying answer:
\[
  \mathbb{WA}(\Circ, \delta) =_{a.e.}
  \begin{cases}
    \emptyset & \mbox{if } \sum \varphi(n)\delta_n < \infty,\\
    \\
    \Circ     & \mbox{if } \sum \varphi(n)\delta_n = \infty.
  \end{cases}
\]
where $\varphi$ is Euler's totient function.

That $\mathbb{WA}(\Circ, \delta) =_{a.e.} \emptyset$ if $\sum \varphi(n)\delta_n < \infty$ is very
easy (it follows from the {\lq}easy{\rq} direction of the Borel-Cantelli theorem). The converse
is extremely hard. It was first stated in 1941 \cite{MR4859} and was one of the most
important open problems in metric number theory for almost 80 years.

A formal proof of the converse would be especially satisfying given how elementary the statement of
the result is. After Gallagher's theorem, perhaps the next best target is lemma 5.2 in
\cite{KoukoulopoulosMaynard2020}.

\subsection{Developing against \texttt{master}}
It would have been impossible to complete the work discussed here without the extensive theories
of algebra, measure theory, topology etc. contained within \mathlib. As we have said, all of our
code was added directly to the \texttt{master} branch of \mathlib; most of it is
{\lq}library code{\rq}, not specific to Gallagher's theorem.

Although it is harder to develop this way, we believe it is essential in order to permit
formalisation of contemporary mathematics. We therefore wish to exhibit this project as further
evidence that this workflow can succeed, and we hope to encourage even more people to follow suit.

\bibliography{paper}

\end{document}